\documentclass[letterpaper, 10pt, conference]{ieeeconf}
\IEEEoverridecommandlockouts
\usepackage{cite}
\usepackage{amsmath,amssymb,amsfonts}
\usepackage{nicefrac}
\usepackage{algorithmic}
\usepackage{graphicx}
\usepackage{algorithm,algorithmic}
\usepackage{textcomp}
\usepackage{amsthm}
\usepackage{soul}
\usepackage{graphics} % for pdf, bitmapped graphics files
\usepackage{epsfig} % for postscript graphics files
\usepackage{mathptmx} % assumes new font selection scheme installed
\usepackage{times} % assumes new font selection scheme installed
\usepackage[acronym]{glossaries}
\usepackage{cleveref}
\usepackage{graphicx}
\usepackage{subcaption}
\usepackage{pstricks} 				% Plotting package
\usepackage{pgfplots} 				% Plotting package
\usepackage{pgfplotstable} 			% Plotting package
\usepackage{tikz}
\usepackage{nicematrix}
\usepackage{commath}
\usepackage{algorithm}
\usepackage{mathtools}
\usepackage{siunitx}
\usepackage{scalerel}
\usepackage{relsize, exscale}
\usepackage{nccmath}
\usepackage{booktabs}
\usepackage{multirow}
\usepackage{bm}
\usepackage{aliascnt}
\usepackage{fancyhdr}
\def\BibTeX{{\rm B\kern-.05em{\sc i\kern-.025em b}\kern-.08em
    T\kern-.1667em\lower.7ex\hbox{E}\kern-.125emX}}
\newacronym{LTI}{LTI}{linear time-invariant}
\newacronym{SDP}{SDP}{semidefinite program}
\newacronym{MIP}{MIP}{mixed-integer program}
\newacronym{IQC}{IQC}{integral quadratic constraint}
\newacronym{SOS}{SOS}{sum of squares}
\newacronym{UAV}{UAV}{unmanned aerial vehicle}
\newacronym{ROA}{RoA}{region of attraction}
\newacronym{ReLU}{ReLU}{Rectified Linear Unit}
\newacronym{REN}{REN}{recurrent equilibrium network}
\newacronym{RNN}{RNN}{recurrent neural network}
\newacronym{GAS}{GAS}{globally asymptotically stable}
\newacronym{LAS}{LAS}{locally asymptotically stable}
\newacronym[longplural=Linear Matrix Inequalities]{LMI}{LMI}{linear matrix inequality}
\newacronym{LSTM}{LSTM}{long short-term memory}
\newacronym{LQR}{LQR}{linear-quadratic regulator}
\newacronym{MPC}{MPC}{model predictive control}
\newacronym{DoF}{DoF}{degree of freedom}
\newacronym{RHP}{RHP}{right-half plane}
\newacronym[longplural=equations of motion]{EoM}{EoM}{equation of motion}
\newacronym{CoG}{CoG}{center of gravity}
\newacronym{NNC}{NNC}{neural-network-based controller}
\newacronym{NN}{NN}{neural network}
\newacronym{MIQP}{MIQP}{mixed-integer quadratic program}
\newacronym{MILP}{MILP}{mixed-integer linear program}
\newacronym{INDI}{INDI}{incremental nonlinear dynamic inversion}
\newacronym{VTOL}{VTOL}{vertical take-off and landing}
\newcommand{\realsN}[1]{\ensuremath{\mathbb{R}^{#1}}}

\newcommand{\transpose}{^\top}

\DeclareMathAlphabet{\mathcal}{OMS}{cmsy}{m}{n}
\useRomanappendicesfalse
\usepgfplotslibrary{groupplots} 	
\usepgfplotslibrary{colormaps}
\usetikzlibrary{plotmarks}
\usetikzlibrary{calc}
\usetikzlibrary{shapes.geometric}
\usetikzlibrary{arrows.meta}
\usetikzlibrary{patterns,patterns.meta}
\pgfplotsset{compat=newest}
\pgfplotsset{plot coordinates/math parser=false}
\newlength\figureheight
\newlength\figurewidth
\pgfplotsset{/pgfplots/layers/niceLayers/.define layer set={
		axis background,axis grid,main,axis ticks,axis lines,axis tick labels,axis descriptions,axis foreground
	}{/pgfplots/layers/standard}
}
\pgfplotsset{every axis/.append style={
		set layers=niceLayers,
		tick label style={font=\scriptsize},
		clip marker paths=true,
		line width=1.5pt,
		line cap=round,
		line join=round,
		tick style={semithick, color=black}
}}
\newtheorem{thm}{Theorem}[section]

\newtheorem{cor}[thm]{Corollary}
\newtheorem{defn}[thm]{Definition}

\newtheorem{rem}{Remark}

\crefname{thm}{Theorem}{Theorems}
\crefname{lemma}{Lemma}{Lemmas}
\crefname{prop}{Proposition}{Propositions}
\crefname{cor}{Corollary}{Corollaries}
\crefname{defn}{Definition}{Definitions}
\crefname{conj}{Conjecture}{Conjectures}
\crefname{exmp}{Example}{Examples}
\crefname{rem}{Remark}{Remarks}
\crefname{assume}{Assumption}{Assumptions}
\crefname{equation}{}{} % Remove "eq" before equation references
\Crefname{equation}{Equation}{Equations} % Remove "eq" before equation references
\crefname{figure}{Fig.}{Figs.} % Ensure figures are correctly referenced
\Crefname{figure}{Fig.}{Figs.} % Ensure figures are correctly
\crefname{table}{Table}{Tables} % Ensure figures are correctly referenced
\Crefname{table}{Table}{Tables} % Ensure figures are correctly
\crefname{subfigure}{Fig.}{Figs.} % Ensure figures are correctly referenced
\Crefname{subfigure}{Fig.}{Figs.} % Ensure figures are correctly 
\crefname{section}{Section}{Sections} % Ensure sectiions are referenced are correctly 
\crefname{algorithm}{Algorithm}{Algorithms} % Ensure algorithms are referenced are correctly
\Crefname{algorithm}{Algorithm}{Algorithms} % Ensure algorithms are referenced are correctly
\crefname{appendix}{Appendix}{Appendices}
\Crefname{appendix}{Appendix}{Appendices}  % Ensure appendix is referenced correctly by definng an alias
\DeclareSIUnit{\deg}{deg}
\fancypagestyle{firstpage}{
  \fancyhf{}
  
  \fancyhead[C]{\footnotesize This work has been submitted to the European Control Conference (ECC) for possible publication. Copyright may be transferred without notice, after which this version may no longer be accessible.}
}
\begin{document}
\title{SOS-based Stability Verification for Saturated INDI Control of Hybrid-VTOL Aircraft Pitch Rate Dynamics}
\author{Dalim Wahby$^\dagger$, Lorenzo Schenk$^\dagger$, Jimmy Vu, Guillaume Ducard, \IEEEmembership{Senior Member, IEEE}
\thanks{$^\dagger$Dalim Wahby and Lorenzo Schenk contributed equally to this work.}%
\thanks{D. Wahby, J. Vu, and G. Ducard are with Universit{\'e} C\^{o}te d`Azur I3S CNRS, 06903 Sophia Antipolis, France. (E-mail: dalim.wahby@univ-cotedazur.fr; jimmy.vu@univ-cotedazur.fr; guillaume.ducard@univ-cotedazur.fr)}%
\thanks{L. Schenk is with the Institute for Dynamic Systems and Control (IDSC), Department of Mechanical and Process Engineering, Swiss Federal Institute of Technology (ETH) Zurich, Leonhardstrasse 21, 8092 Zurich, Switzerland. (E-mail: lschen@ethz.ch)}%
\thanks{This work was supported by the French government through the France 2030 investment plan managed by the National Research Agency (ANR), as part of the Initiative of Excellence Universit{\'e} C\^{o}te d`Azur under reference number ANR- 15-IDEX-01.}
}

\maketitle
\thispagestyle{firstpage}
\begin{abstract}
\Acrfull{INDI} is a prominent flight-control strategy valued for its robust disturbance rejection; however, its formal stability verification has traditionally been limited to linearized dynamical models. This paper presents a formal nonlinear stability certificate for a saturated \acrshort{INDI} pitch-rate controller for an open-loop unstable hybrid \acrfull{VTOL} aircraft by representing the \acrshort{INDI} controller via an equivalent \acrfull{REN}. By casting the saturated \acrshort{INDI} architecture as a \acrshort{REN}, the closed-loop dynamics are exactly mapped to an augmented state-feedback system. This structural equivalence enables the use of \acrfull{SOS} programming to synthesize a locally valid Lyapunov function without relying on conservative bounding approximations. The resulting certificate yields an inner estimate of the \acrfull{ROA} that explicitly accounts for actuator saturation, formally verifying the controller's stability in operating regimes where standard linear margins lose their validity. 
\end{abstract}

\section{Introduction}
\label{sec:introduction}
Incremental nonlinear dynamic inversion (\acrshort{INDI}) has become a widely known flight-control strategy for \acrfull{VTOL} \acrfullpl{UAV} \cite{smeur2016adaptive, Schlatter2024}, as depicted in \cref{fig:vtol_uav}. By commanding actuator increments derived from onboard angular-acceleration measurements and a local control-effectiveness estimate, \acrshort{INDI} avoids modeling the feedback nonlinearity of classical nonlinear dynamic inversion (NDI) \cite{ducard2008ndi1, ducard2008ndi2}, reducing sensitivity to aerodynamic uncertainty while retaining strong disturbance-rejection properties \cite{park2025adaptive}, \cite{Rotaetal2026}. 
Closed-loop stability of such inner loops has mainly been characterized through linear, linearized, or sampled-data analyses \cite{Schenk2026,veld2018stability,Lu2020}. 
However, when the pitch-rate loop is excited near its authority limit, the elevon deflection saturates, causing the commanded and applied signals to decouple \cite{Schlatter2024}. A formal nonlinear stability certificate that explicitly encompasses this saturation is therefore required to characterize the true local stability properties of the \acrshort{INDI} pitch-rate loop. However, to the best of the authors' knowledge, a stability certificate of this type has yet to be reported.

The \acrshort{INDI} requires internal states to realize the low-pass filtering of angular acceleration and control signals. When combined with actuator saturation, which is shown in \cite{Korda2022} to decompose exactly into two \acrfull{ReLU} functions, the resulting controller architecture becomes structurally equivalent to a \acrfull{RNN}. By casting this saturated \acrshort{INDI} architecture as a \acrfull{REN}, the stability verification problem can be addressed via tools recently developed for the verification of \acrshortpl{NNC}.

One prominent class of verification methods relies on \acrfull{MIP} and bound-propagation techniques to establish stability by formally verifying reachability \cite{Rober2023, Kotha2023} or Lyapunov conditions \cite{Zhou2022, Chang2019, Wu2023, Yang2024} over bounded continuous domains. While these approaches scale well with the size of the neural network, their reliance on state-space partitioning and branch-and-bound search causes computational complexity to scale unfavorably with the system's state and recurrent dimensions. As the saturated \acrshort{INDI} architecture requires an augmented state space to account for its internal filters and estimators, applying these partitioning methods would introduce an unnecessary computational cost.

A second family of verification tools consists of \acrshort{SDP}-based methods. While many such approaches rely on sector or quadratic constraints \cite{Yin2022, Richardson2023} to approximate and bound nonlinearities, a specific subset utilizes semialgebraic sets to represent the system's exact input-output relations \cite{Korda2022, Detailleur2025}. This framework has been proven compatible with the broad class of \acrshortpl{REN} \cite{Detailleur2025}, which encompasses \acrshortpl{RNN} \cite{Revay2024}. As the \acrshort{ReLU} elements representing actuator saturation admit a semialgebraic description, the exact input-output properties of the recurrent architecture can be captured entirely through polynomial (in)equalities. Crucially, \acrfull{SOS}-based methods have been shown to scale to closed-loop systems comprising combined totals of hundreds of states, hidden activation functions, and control inputs \cite{Detailleur2025, Detailleur2025Sigi, Korda2022}. By formulating the stability verification problem via \acrshort{SOS} programming, the proposed methodology eliminates bounding conservatism while easily scaling to the augmented state dimension of the \acrshort{INDI} architecture.

\begin{figure}[t]
    \centering
    \includegraphics[width=0.8\linewidth]{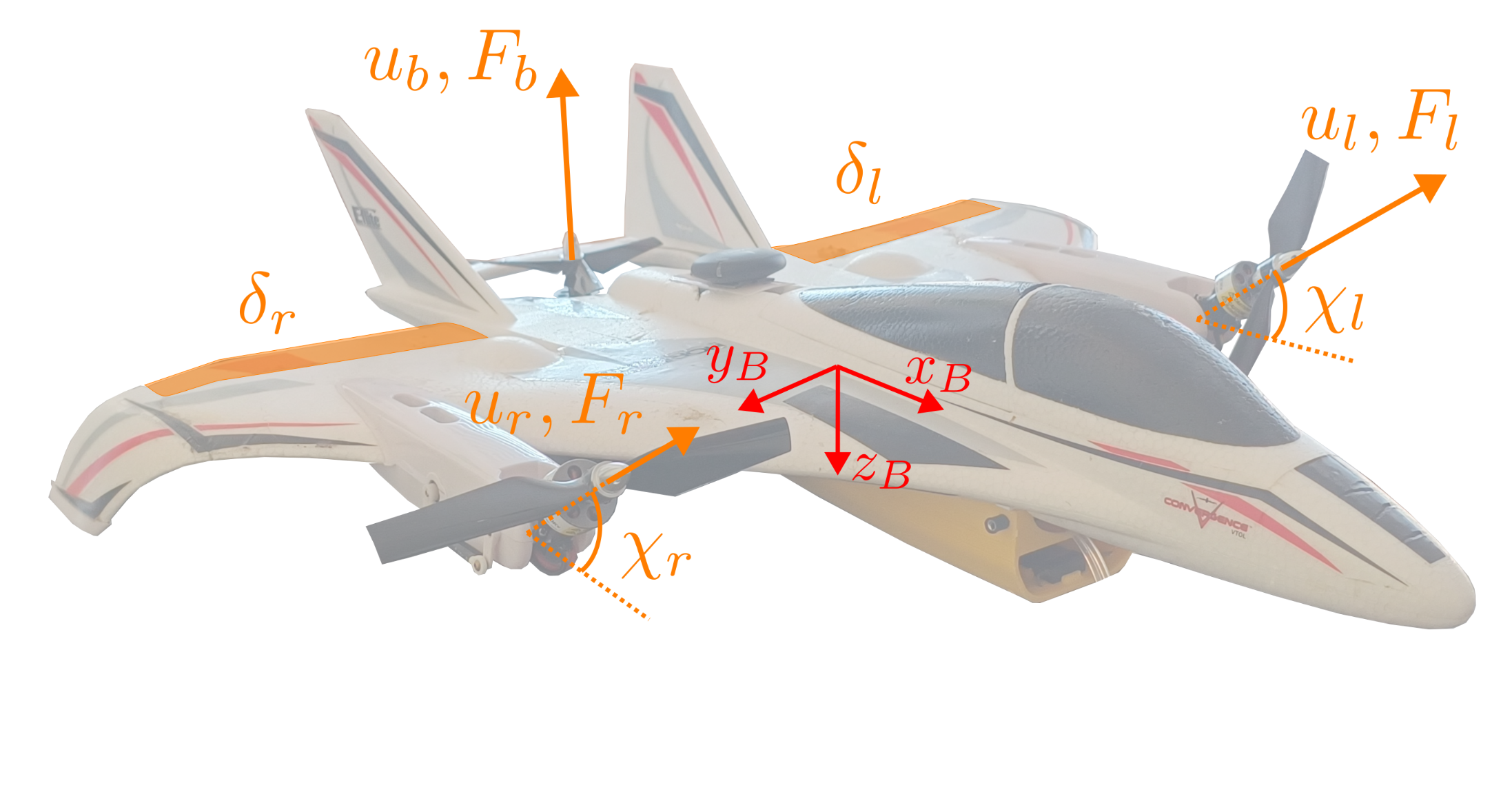}
    \caption{Schematic of the VTOL UAV \cite{Schlatter2024}}
    \label{fig:vtol_uav}
\end{figure}

\subsection{Contributions}
\label{sec:Introduction_Contributions}
While \acrshort{INDI} pitch-rate controllers provide robust and good performance in practice, formally verifying their stability in the presence of nonlinearities, such as actuator saturation, remains an unsolved problem. This paper leverages semialgebraic-set-based stability verification procedure of \cite{Korda2022, Detailleur2025} to the saturated \acrshort{INDI} architecture proposed in \cite{Schlatter2024}. The specific contributions are:
\begin{enumerate}
    \item \textbf{\acrshort{REN} modeling of the \acrshort{INDI} controller and state-estimator:} The saturated \acrshort{INDI} controller, including the state-estimator, is formally proven to be equivalent to a \acrshort{REN}, a class of \acrshortpl{NN} that includes \acrshortpl{RNN}.
    \item \textbf{Augmented closed-loop formulation:} The feedback interconnection of the system and the \acrshort{REN} is reformulated as an augmented system controlled by a static, saturated state-feedback operator. This structural equivalence is leveraged to initialize the stability verification procedure and define the search region.
    \item \textbf{Local stability certificate for the closed-loop system under actuator saturation:} Semialgebraic-set-based stability verification is utilized to synthesize a locally valid Lyapunov function, which provides a certified inner estimate of the \acrfull{ROA} that reaches into the saturated region of the state-space.
\end{enumerate}

The remainder of this paper is organized as follows. \Cref{sec:SystemModelandControllerArchitecture} presents the aircraft pitch model, the \acrshort{INDI} controller architecture, and the state estimator. \Cref{sec:SOS-basedStabilityVerification} details the stability verification procedure and its instantiation to the present system. Numerical results are reported in \cref{sec:Results}, and \cref{sec:Conclusion} concludes this paper with a brief outlook on future work.

\subsection{Notation}
\label{sec:Introduction_Notation}
This work uses the following notational conventions:
\begin{itemize}
    \item In the analysis of discrete-time systems, the plus superscript indicates the successor variable, e.g. $x, \ x^+ \in \realsN{n_x}$ represent the current and successor state, respectively.
    \item All inequalities are defined element-wise. $P \succ 0$, $P \succeq 0$ denote a positive definite and positive semidefinite matrix $P$, respectively.
    \item The notation $\mathcal{M}(x, \,n)$ denotes the vector of all unique products of $n$ entries of $x$, e.g. $\mathcal{M}(\bigl[ \begin{smallmatrix} x_1 \\ x_2 \end{smallmatrix} \bigr], \,2) = [x_1^2, \, x_1x_2, \, x_2^2]\transpose$.
    \item The identity matrix is denoted $I$.
    \item A continuous function $\alpha \colon \mathbb{R}_{\ge 0} \to \mathbb{R}_{\ge 0}$ is said to belong to class $\mathcal{K}$ if it is strictly increasing and $\alpha(0) = 0$. Furthermore, $\alpha$ belongs to class $\mathcal{K}_\infty$ if it belongs to class $\mathcal{K}$ and is radially unbounded, i.e., $\lim_{x \to \infty} \alpha(x) = \infty$.
    \item A set $\mathcal{X} \subseteq \mathbb{R}^{n_{x}}$ is said to be positively invariant with respect to the system dynamics if $ x \in \mathcal{X}$ implies $f\big(x, \varphi(x)\big) \in \mathcal{X}$.
    \item Unless specified otherwise, all norms represent the Euclidean norm.
\end{itemize}

\section{System Model and Controller Architecture}
\label{sec:SystemModelandControllerArchitecture}
In this section, the mathematical models for the aircraft system, the control architecture, and the state estimator are presented, as shown in \cref{fig:pitch_rate_controller}. Subsequently, the actuator saturation is modeled through an exact \acrshort{ReLU} decomposition, which is required for the subsequent stability verification.
\begin{figure*}[t]
    \centering
    \includegraphics[width=\linewidth]{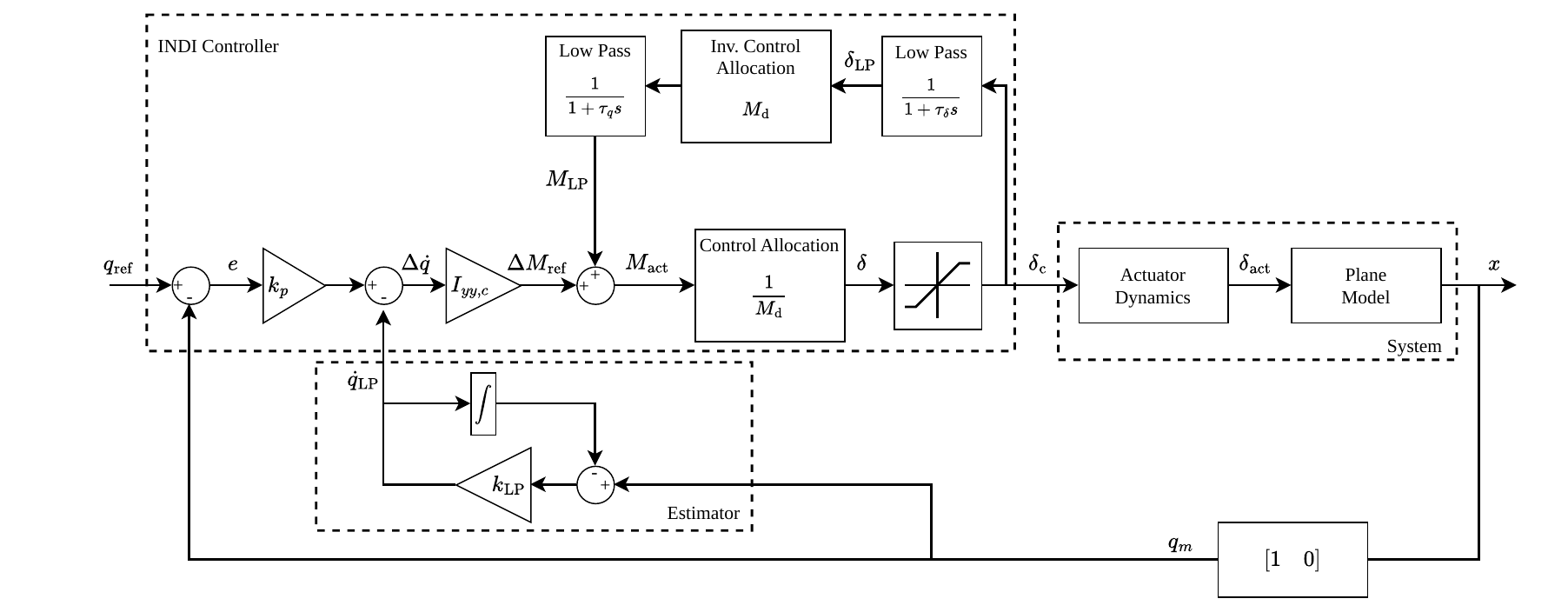}
    \caption{Schematic of the \acrshort{INDI} pitch rate controller in closed-loop with the system.}
    \label{fig:pitch_rate_controller}
\end{figure*}

\subsection{Aircraft Pitch Model}
\label{sec:SystemModel}
Linearized longitudinal dynamics about a trimmed flight condition reduce to the single-axis pitch \acrfull{EoM}
\begin{equation}
  I_{yy,m}\,\dot{q}_{\textrm{m}} = M_m\,\delta_{\textrm{act}} - d_q\,q_{\textrm{m}} + M_{\alpha}\,\alpha,
  \label{eq:pitchEOM}
\end{equation}
where $q_{\textrm{m}}$ is the pitch rate, $M_m$ the pitch-moment effectiveness, $d_q > 0$ the aerodynamic pitch-rate damping coefficient, $I_{yy,m}$ the true moment of inertia, $M_\alpha$ is the moment contribution of the fuselage on the pitch axis and $\alpha$ is the angle of attack. The actuator dynamics are modeled as a first-order lag with time constant
$\tau_{\textrm{act}}$
\begin{equation}
  \dot\delta_{\textrm{act}} = \frac{1}{\tau_{\textrm{act}}}
                     \bigl(\delta_{\textrm{c}} - \delta_{\textrm{act}}\bigr),
  \label{eq:actuatorDynamics}
\end{equation}
where $\delta_{\textrm{c}} \in [\delta_{\textrm{min}}, \delta_{\textrm{max}}]$ is the commanded, saturated elevon deflection. Since $\delta_{\textrm{act}}$ is uniquely determined by $q_{\textrm{m}}$ and $\dot q_{\textrm{m}}$ through \cref{eq:pitchEOM}, it can be eliminated algebraically. For small angle of attack values, it follows that $\dot \alpha \approx q_m$, thus, the resulting continuous-time state-space representation with state $x = [q_{\textrm{m}},\, \dot q_{\textrm{m}}]\transpose \in \realsN{2}$, is given by
\begin{subequations}
\label{eq:continuousPlantGroup}
    \begin{align}
        \dot x &= A_{p,c}\,x + B_{p,c}\,\delta_{\textrm{c}}, \label{eq:continuousPlant} \\[8pt]
        A_{p,c} &= \begin{bmatrix}
            0 & 1 \\[4pt]
            \dfrac {M_{\alpha}}{I_{yy,m}}-\dfrac{d_q}{I_{yy,m}\tau_{\textrm{act}}} &
            -\dfrac{d_q}{I_{yy,m}} - \dfrac{1}{\tau_{\textrm{act}}}
        \end{bmatrix}, \\
        B_{p,c} &= \begin{bmatrix} 0 \\[4pt]
            \dfrac{M_m}{I_{yy,m}\tau_{\textrm{act}}} \end{bmatrix}. \label{eq:plantMatrices}
    \end{align}
\end{subequations}
An unstable configuration is selected by setting a destabilizing pitch-moment derivative $M_\alpha$. The corresponding nominal values are listed in \cref{tab:nominalParameters}. 
Finally, a forward-Euler discretization at sampling period $T_s$ gives the discrete-time matrices $A_p = I + T_s A_{p,c}$ and $B_p = T_s B_{p,c}$.

\subsection{INDI Pitch-Rate Controller and State Estimator}
\label{sec:ControllerArchitecture}
The \acrshort{INDI} pitch-rate controller used in this work consists of four interconnected blocks: a proportional error path, an inertia-scaled moment reference, a moment estimator, and an inverse control-allocation block. The internal states of the \acrshort{INDI} controller are defined as follows
\begin{subequations}
\label{eq:controllerODEs}
\begin{align}
  \dot\delta_{\textrm{LP}} &= \frac{1}{\tau_\delta}
                     \bigl(\delta_{\textrm{c}} - \delta_{\textrm{LP}}\bigr),
                   \label{eq:deltaLP}\\
  \dot M_{\textrm{LP}}     &= \frac{1}{\tau_q}
                     \bigl(M_d\,\delta_{\textrm{LP}} - M_{\textrm{LP}}\bigr),
                   \label{eq:MLP}
\end{align}
\end{subequations}
where $\tau_\delta$ and $\tau_q$ are first-order filter time constants, and $M_d$ is the assumed control effectiveness used in both the control and inverse control allocation. Note that $\delta_{\textrm{c}}$, not the pre-saturation command $\delta$, drives \cref{eq:deltaLP}, thus, the estimator reconstructs the moment that was applied to the system.

To obtain a pitch-rate-estimate a low-pass filter is used. The resulting rate-estimate dynamics can be expressed as follows
\begin{equation}
    \dot q_{\textrm{LP}}     = k_{\textrm{LP}}\bigl(q_{\textrm{m}} - q_{\textrm{LP}}\bigr), 
    \label{eq:qLP}
\end{equation}
where $k_{\textrm{LP}}$ is the corner frequency of the pitch-rate-estimator. Additionally, the nominal time-alignment condition $k_{\textrm{LP}} = 1/\tau_q$, $\tau_\delta = \tau_{\textrm{act}}$ ensures phase consistency between the estimator and the actuator \cite{Schenk2026, veld2018stability}.

The incremental acceleration demand is $\Delta\dot q = k_p(q_\textrm{ref} - q_{\textrm{m}}) - \dot q_{\textrm{LP}}$, where $\dot q_{\textrm{LP}}$ is the low-pass angular-acceleration estimate obtained from \cref{eq:qLP}, $q_\textrm{ref}$ is the reference pitch and $k_p$ is the proportional gain. The moment reference $\Delta M_\textrm{ref} = I_{yy,c}\,\Delta\dot q$ is formed using the controller inertia estimate $I_{yy,c}$.
The total moment command $M_\textrm{act} = \Delta M_\textrm{ref} + M_{\textrm{LP}}$ is then inverted to give the pre-saturation deflection command
\begin{equation}
  \delta = \frac{I_{yy,c}[k_p(q_{\textrm{ref}}-q_{\textrm{m}})-\dot{q}_{\textrm{LP}}] + M_{\textrm{LP}}}{M_d} \;.
  \label{eq:preSatCommand}
\end{equation}

The elevon deflection is hard-limited to $\delta_{\textrm{c}} \in [\delta_{\textrm{min}}, \delta_{\textrm{max}}]$, thus, the saturation should be modeled to be taken into account in the closed-loop dynamics. Following \cite{Korda2022, Detailleur2025}, the saturation is decomposed
\emph{exactly} into the composition of two \acrshort{ReLU} functions, hence, it holds for all $\delta \in \mathbb{R}$
\begin{subequations}
\label{eq:ReLU}
    \begin{align}
      \lambda_1 & = \textrm{ReLU}(\delta - \delta_{\textrm{min}}), \label{eq:lambda_1}\\
      \lambda_2 & = \textrm{ReLU}(\delta - \delta_{\textrm{max}}), \label{eq:lambda_2}\\
      \delta_{\textrm{c}} & = \delta_{\textrm{min}} + \lambda_1 - \lambda_2,
      \label{eq:satReLU}
    \end{align}
\end{subequations}
where $\textrm{ReLU}(x) = \max(0,x)$.
This control law has been validated in simulation \cite{Schlatter2024}; however, formal stability analyses have until now been restricted to the linearized system \cite{Schenk2026}, where the saturation is neglected.

\section{SOS-Based Stability Verification}
\label{sec:SOS-basedStabilityVerification}
In this section, a formal procedure for verifying the local stability of the closed-loop system with respect to the origin, i.e. $q_{\textrm{ref}}=0$, is presented. The methodology is structured into two primary phases. First, the closed-loop dynamics are transformed into a semialgebraic system model to satisfy the requirements of \acrshort{SOS} programming. Second, a sequence of \acrshortpl{SDP} is formulated to synthesize a locally valid Lyapunov function, providing a certified inner estimate of the \acrshort{ROA}.

\subsection{SOS Compatible System Model}
The system must be expressed via semialgebraic sets, i.e. sets defined exclusively by polynomial (in)equalities. The derivation of this exact format for the saturated \acrshort{INDI} architecture is detailed here. The process begins with establishing the equivalence of the \acrshort{INDI} controller to a \acrshort{REN}, which is then integrated into a normalized, augmented closed-loop system representation. This model is subsequently used to derive the specific polynomial constraints for the system nonlinearities, yielding the formal mathematical structure required for the stability verification.
\begin{figure}
    \centering
    \includegraphics[width=\linewidth]{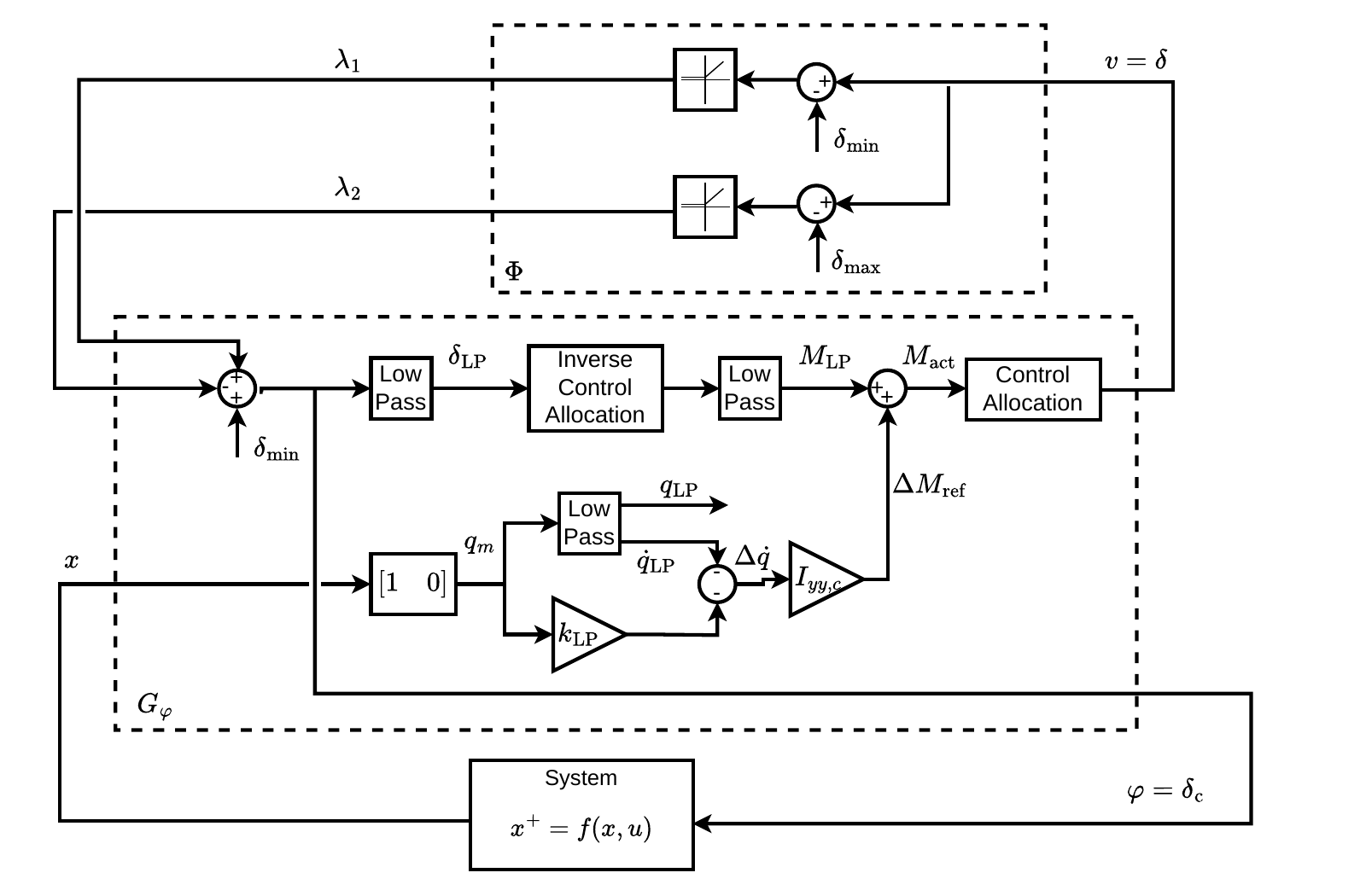}
    \caption{Equivalent representation of the \acrshort{INDI} controller with isolated nonlinearity and linear component. As the stability analysis is conducted with respect to the origin, i.e. $q_{\textrm{ref}}=0$, $q_{\textrm{ref}}$ is omitted in this schematic.}
    \label{fig:IsolatedNonliearity}
\end{figure}

\subsubsection{Equivalent REN \& Augmented System}
\label{sec:EquivalentRENandAugmentedSystem}
In their most general form, \acrshortpl{REN} are defined by a feedback interconnection consisting of a linear system $G_{\varphi}$ and a memoryless nonlinear operator $\phi$ \cite{Revay2024}, as shown in the dashed box in \cref{subfig:PreRENTransformation}.
Letting the inputs, outputs and internal state variable associated with linear system $G_{\varphi}$ be denoted by $\bigl[\!\begin{smallmatrix} \lambda \\ x \end{smallmatrix}\!\bigr]$, $\bigl[\!\begin{smallmatrix} v \\ \varphi \end{smallmatrix}\!\bigr]$ and $x_{\varphi}$, respectively, a general \acrshort{REN} is described by
\begin{subequations}
    \label{eq:RENdescription}
    \begin{gather}
        \label{eq:RENSSdescription}
        \begin{bNiceArray}{c}
            x_{\varphi}^+ \\ %[1.5pt] 
            v \\
            \varphi            
        \end{bNiceArray}
        =
        \begin{bNiceArray}{c|c@{\hskip 5pt}c}
            A & B_{1} & B_{2} \\[1pt] 
            C_{1} & D_{11} & D_{12} \\
            C_{2} & D_{21} & D_{22} 
            \CodeAfter
            \tikz \draw [transform canvas={yshift=1pt}, shorten > = 0.35em, shorten < = 0.35em](2-|1) -- (2-|last) {};
        \end{bNiceArray}
        \begin{bNiceArray}{c}
            x_{\varphi} \\ %[1pt] 
            \lambda \\
            x            
        \end{bNiceArray}
        +
        \begin{bNiceArray}{c}
            b_{x_{\varphi}} \\ %[1pt] 
            b_{v} \\
            b_{\varphi}
        \end{bNiceArray},
        \\
        \label{eq:RENActivationFunction}
        \lambda_i = \phi_i(v_i) \quad \forall i \in [n_\lambda],
    \end{gather}
\end{subequations}
where $n_\lambda$ represents the number of neurons in the \acrshort{REN}.
\begin{figure*}[t]
    \centering
    \begin{subfigure}{0.35\textwidth} 
        \centering
         %% Creator: Inkscape 1.1 (c68e22c387, 2021-05-23), www.inkscape.org
%% Accompanies image file 'RENTransformation_Before.pdf' (pdf, eps, ps)
\begingroup%
  \makeatletter%
  \providecommand\color[2][]{%
    \errmessage{(Inkscape) Color is used for the text in Inkscape, but the package 'color.sty' is not loaded}%
    \renewcommand\color[2][]{}%
  }%
  \providecommand\transparent[1]{%
    \errmessage{(Inkscape) Transparency is used (non-zero) for the text in Inkscape, but the package 'transparent.sty' is not loaded}%
    \renewcommand\transparent[1]{}%
  }%
  \providecommand\rotatebox[2]{#2}%
  \newcommand*\fsize{\dimexpr\f@size pt\relax}%
  \newcommand*\lineheight[1]{\fontsize{\fsize}{#1\fsize}\selectfont}%
  \ifx\svgwidth\undefined%
    \setlength{\unitlength}{195.70561451bp}%
    \ifx\svgscale\undefined%
      \relax%
    \else%
      \setlength{\unitlength}{\unitlength * \real{\svgscale}}%
    \fi%
  \else%
    \setlength{\unitlength}{\svgwidth}%
  \fi%
  \global\let\svgwidth\undefined%
  \global\let\svgscale\undefined%
  \makeatother%
  \begin{picture}(1,0.63327518)%
    \lineheight{1}%
    \setlength\tabcolsep{0pt}%
    \put(0,0){\includegraphics[width=\unitlength,page=1]{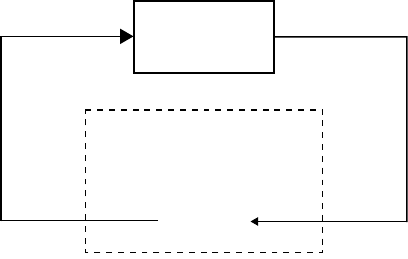}}%
    \put(0.38344241,0.53285429){\makebox(0,0)[lt]{\lineheight{1.25}\smash{\begin{tabular}[t]{l}$x^+ = f(x,u)$\end{tabular}}}}%
    \put(0.05577581,0.12232683){\makebox(0,0)[lt]{\lineheight{1.25}\smash{\begin{tabular}[t]{l}$u(x)$\end{tabular}}}}%
    \put(0.7519792,0.56566617){\makebox(0,0)[lt]{\lineheight{1.25}\smash{\begin{tabular}[t]{l}$x$\end{tabular}}}}%
    \put(0,0){\includegraphics[width=\unitlength,page=2]{RENTransformation_Before.pdf}}%
    \put(0.48412852,0.24785409){\makebox(0,0)[lt]{\lineheight{1.25}\smash{\begin{tabular}[t]{l}$\phi$\end{tabular}}}}%
    \put(0.47730358,0.10231663){\makebox(0,0)[lt]{\lineheight{1.25}\smash{\begin{tabular}[t]{l}{\large $G_\varphi$}\end{tabular}}}}%
    \put(0.23367888,0.19787366){\makebox(0,0)[lt]{\lineheight{1.25}\smash{\begin{tabular}[t]{l}$v$\end{tabular}}}}%
    \put(0.71374434,0.19787366){\makebox(0,0)[lt]{\lineheight{1.25}\smash{\begin{tabular}[t]{l}$\lambda$\end{tabular}}}}%
  \end{picture}%
\endgroup%

        \caption{}
        \label{subfig:PreRENTransformation}
    \end{subfigure}
    \hspace{0.15\textwidth}
    \begin{subfigure}{0.35\textwidth} 
      \centering
      %% Creator: Inkscape 1.1 (c68e22c387, 2021-05-23), www.inkscape.org
%% Accompanies image file 'RENTransformation_After.pdf' (pdf, eps, ps)
\begingroup%
  \makeatletter%
  \providecommand\color[2][]{%
    \errmessage{(Inkscape) Color is used for the text in Inkscape, but the package 'color.sty' is not loaded}%
    \renewcommand\color[2][]{}%
  }%
  \providecommand\transparent[1]{%
    \errmessage{(Inkscape) Transparency is used (non-zero) for the text in Inkscape, but the package 'transparent.sty' is not loaded}%
    \renewcommand\transparent[1]{}%
  }%
  \providecommand\rotatebox[2]{#2}%
  \newcommand*\fsize{\dimexpr\f@size pt\relax}%
  \newcommand*\lineheight[1]{\fontsize{\fsize}{#1\fsize}\selectfont}%
  \ifx\svgwidth\undefined%
    \setlength{\unitlength}{195.70527435bp}%
    \ifx\svgscale\undefined%
      \relax%
    \else%
      \setlength{\unitlength}{\unitlength * \real{\svgscale}}%
    \fi%
  \else%
    \setlength{\unitlength}{\svgwidth}%
  \fi%
  \global\let\svgwidth\undefined%
  \global\let\svgscale\undefined%
  \makeatother%
  \begin{picture}(1,0.63338911)%
    \lineheight{1}%
    \setlength\tabcolsep{0pt}%
    \put(0,0){\includegraphics[width=\unitlength,page=1]{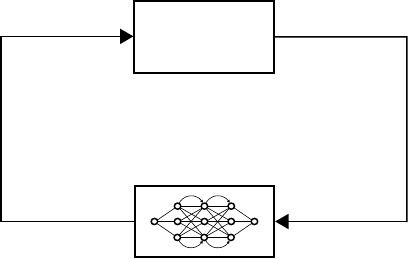}}%
    \put(0.36307978,0.5327426){\makebox(0,0)[lt]{\lineheight{1.25}\smash{\begin{tabular}[t]{l}$\tilde{x}^+ = \tilde{f}(\tilde{x},\tilde{\varphi})$\end{tabular}}}}%
    \put(0.05376934,0.12024305){\makebox(0,0)[lt]{\lineheight{1.25}\smash{\begin{tabular}[t]{l}$\tilde{\varphi}(\tilde{x})$\end{tabular}}}}%
    \put(0.75198121,0.56578001){\makebox(0,0)[lt]{\lineheight{1.25}\smash{\begin{tabular}[t]{l}$\tilde{x}$\end{tabular}}}}%
  \end{picture}%
\endgroup%

      \caption{}
      \label{subfig:PostRENTransformation}
    \end{subfigure}
    \caption{Summarized block diagrams of \cref{fig:IsolatedNonliearity} are shown in  \subref{subfig:PreRENTransformation}: an open-loop system $x^+ = f(x,u)$, where $u(x) = \varphi(x)$ in closed-loop with a \acrshort{REN}-based \acrshort{NNC} shown via fractional transformation in the dotted box, and \subref{subfig:PostRENTransformation} the equivalent system obtained by augmenting the state of and input to the dynamical system with the \acrshort{REN}'s internal state variable $x_\varphi$ and hidden variable $\lambda$, respectively \cite{Detailleur2025}.}
    \centering
    \label{fig:RENTransformation}
\end{figure*}

\begin{thm}[REN equivalence and well-posedness]
    \label{thm:RENEquivalence}
    Consider the saturated \acrshort{INDI} controller and estimator defined by \cref{eq:controllerODEs,eq:qLP,eq:preSatCommand,eq:ReLU}. The discrete-time, well-posed, equivalent \acrshort{REN} formulation, defined by  $x = [q_{\textrm{m}}, \dot{q}_{\textrm{m}}]\transpose \in \mathbb{R}^2$ and internal state vector $x_\varphi = [q_{\textrm{LP}}, \delta_{\textrm{LP}}, M_{\textrm{LP}}]\transpose \in \mathbb{R}^3$ and, hidden activation functions $\lambda = [\lambda_1, \lambda_2]\transpose \in \mathbb{R}^2$ as defined in \cref{eq:lambda_1}, \cref{eq:lambda_2}, and output $\varphi = \delta_{\textrm{c}}$, is given by
    \begin{subequations}
    \label{eq:REN_matrices}
    \begin{align}
        A &= \begin{bmatrix}
            1 - T_s k_{\textrm{LP}} & 0 & 0 \\[3pt]
            0 & 1 - \tfrac{T_s}{\tau_\delta} & 0 \\[3pt]
            0 & \tfrac{T_s M_d}{\tau_q} & 1 - \tfrac{T_s}{\tau_q}
        \end{bmatrix}, \label{eq:RENA} \\[6pt]
        B_1 &= \begin{bmatrix}
            0 & 0 \\[3pt]
            \tfrac{T_s}{\tau_\delta} & -\tfrac{T_s}{\tau_\delta} \\[3pt]
            0 & 0
        \end{bmatrix}, \quad
        B_2 = \begin{bmatrix}
            T_s k_{\textrm{LP}} & 0 \\[3pt]
            0 & 0 \\[3pt]
            0 & 0
        \end{bmatrix}, \label{eq:RENB} \\[6pt]
        C_1 &= \begin{bmatrix}
            \dfrac{I_{yy,c} k_{\textrm{LP}}}{M_d} & 0 & \dfrac{1}{M_d} \\
            \dfrac{I_{yy,c} k_{\textrm{LP}}}{M_d} & 0 & \dfrac{1}{M_d}
        \end{bmatrix}, \label{eq:RENC} \\[6pt]
        D_{12} &= \begin{bmatrix}
            -\dfrac{I_{yy,c}(k_p + k_{\textrm{LP}})}{M_d} & 0 \\
             -\dfrac{I_{yy,c}(k_p + k_{\textrm{LP}})}{M_d} & 0
        \end{bmatrix}, \quad
        D_{21} = \begin{bmatrix} 1 & -1 \end{bmatrix}, \label{eq:REND} \\[6pt]
        b_{x_\varphi} &= \begin{bmatrix}
            0 \\[3pt]
            \tfrac{T_s\delta_{\textrm{min}}}{\tau_\delta} \\[3pt]
            0
        \end{bmatrix}, \quad
        \begin{aligned}
            b_v &= [-\delta_{\textrm{min}}, -\delta_{\textrm{max}}]\transpose, \\
            b_\varphi &= \delta_{\textrm{min}},
        \end{aligned} \label{eq:REN_biases}
    \end{align}
    \end{subequations}
    where the remaining matrices are $D_{11} = 0_{2 \times 2}$, $D_{22} = 0_{1 \times 2}$, and $C_2 = 0_{1 \times 3}$.
\end{thm}
\begin{proof}
First, the nonlinearity of the \acrshort{INDI} controller is isolated, as shown in \cref{fig:IsolatedNonliearity}. Then by substituting \cref{eq:satReLU} into \cref{eq:deltaLP}, $B_1$ is obtained. By additionally applying forward-Euler discretization to \cref{eq:controllerODEs,eq:qLP} with sampling period $T_s$, we get \cref{eq:RENA}, \cref{eq:RENB}, and $b_{x_\varphi}$. Rearranging \cref{eq:preSatCommand} and \cref{eq:satReLU} gives \cref{eq:RENC}, \cref{eq:REND}, and \cref{eq:REN_biases}. Therefore, by isolating the linear dynamics of the low-pass filters as the \acrshort{REN}'s linear system $G_\varphi$ and assigning the saturation-related non-linearities to the operator $\phi$, the block-matrix structure in \eqref{eq:RENdescription} is obtained.

Finally, by \cite{Revay2020_EquilibriumNet} the \acrshort{REN} is well-posed if there exists a positive-definite diagonal matrix $\Lambda$, such that
\begin{equation}
    \label{eq:WellPosednessREN}
    2\Lambda - \Lambda D_{11} - D_{11}\transpose\Lambda \succ 0.
\end{equation}
Since $D_{11} = 0_{2 \times 2}$, it can be concluded that any positive-definite diagonal matrix $\Lambda$ satisfies this condition and thus, the \acrshort{REN} is well-posed. This ensures the existence and uniqueness of the internal variables, enabling the application of the stability verification method proposed in \cite{Detailleur2025}.
\end{proof}
\begin{cor}[Augmented saturated state-feedback form]
\label{cor:AugmentedStateFeedbackController}
    The closed-loop system, formed by the feedback interconnection of the system's discrete-time dynamics and the \acrshort{REN}-based controller of \cref{thm:RENEquivalence}, is equivalent to an augmented system defined by state $\tilde x = [x\transpose, x_\varphi\transpose]\transpose$ and a saturated state-feedback controller $\tilde\varphi : \realsN{5}\to \realsN{3},  \tilde\varphi(\tilde x)= [\lambda(\tilde x)\transpose, \varphi(\tilde x)]\transpose$ given by
\begin{subequations}
\label{eq:augmentedDynamicsGroup}
\begin{align}
    & \tilde x^+ = \tilde A\,\tilde x + \tilde B\,\tilde\varphi(\tilde{x})
                 + \begin{bmatrix} 0_{2\times1} \\ b_{x_\varphi} \end{bmatrix}, \label{eq:augmentedSystem} \\[6pt]
    \tilde A &= \begin{bmatrix} A_p & 0_{2\times 3} \\ B_2 & A \end{bmatrix} \in \mathbb{R}^{5 \times 5}, \quad
    \tilde B = \begin{bmatrix} 0_{2\times 2} & B_p \\ B_1 & 0_{3\times 1} \end{bmatrix} \in \mathbb{R}^{5 \times 3}. \label{eq:augmentedMatrices} 
\end{align}
\end{subequations}
\end{cor}

\begin{proof}
    The augmented closed-loop state-space representation is obtained by applying the procedure established in \cite[Sec.~III-B]{Detailleur2025} to the feedback interconnection of the system and the \acrshort{REN} representation of \cref{thm:RENEquivalence}. Note that in the linear region, i.e., within the saturation bounds as defined in \cref{eq:ReLU}, it follows that
    \begin{subequations}
    \begin{align}
                \lambda_1(\tilde x)  &= \operatorname{ReLU}(K\tilde x - \delta_{min}) = K\tilde x - \delta_{min} \\
        \lambda_2(\tilde x) &= 0 \\
        \varphi(\tilde{x}) & = K  \tilde x
    \end{align}
    \end{subequations}
    where $K \in \mathbb{R}^{1 \times 5}$ is defined as
    \begin{equation}
        K = \begin{bmatrix}
            -\frac{I_{yy,c} (k_p + k_{\operatorname{LP}})}{M_d} & 0 & \frac{I_{yy,c} k_{\operatorname{LP}}}{M_d} & 0 & \frac{1}{M_d}
        \end{bmatrix}
    \end{equation} 
Within that same region, we can compute the feedback gain matrix
    \begin{equation}
        \tilde K = \left. \frac{\partial \tilde \varphi}{\partial \tilde x}\right|_{\tilde x=0}= \begin{bmatrix}
                K
                \\ 0 \\
                K
        \end{bmatrix} \in \mathbb{R}^{3 \times 5}
    \end{equation}
    Hence, the resulting augmented controller $\tilde\varphi(\tilde x)$ corresponds to a memoryless nonlinear operator with feedback-gain matrix $\tilde K$ 
    \begin{equation}
        \tilde{\varphi}(\tilde x) = \tilde{K} \tilde x
    \end{equation} 
    Consequently, the evaluation of the hidden activations $\lambda$ and the saturation output $\varphi$ depends solely on the current value of $\tilde{x}$ via the affine mapping
    $v(\tilde x) = \begin{bmatrix}
        K\tilde x , \ K\tilde x 
    \end{bmatrix}\transpose
    + b_v$.
Thus, the dynamic feedback interconnection is transformed into a static nonlinearity.
\end{proof}

 As proposed in \cite{Detailleur2025Sigi}, prior to defining the final semialgebraic sets and formulating the \acrshortpl{SDP}, all states are normalized to ensure numerical conditioning of the optimization problems. The normalized coordinates $\bar{x} = D_x^{-1}\tilde{x}$ are obtained using the diagonal scaling matrix
\begin{equation}
  D_x = \mathrm{diag}\!\bigl(
    s_{q_{\textrm{m}}},\; s_{\dot{q}_{\textrm{m}}},\; s_{q_{\textrm{LP}}},\;
    s_{\delta_{\textrm{LP}}},\; s_{M_{\textrm{LP}}}\bigr),
  \label{eq:scalingMatrix}
\end{equation}
where each entry reflects the characteristic physical magnitude of the corresponding state variable: $s_{q_{\textrm{m}}} = 6~\text{rad s$^{-1}$}$, $s_{\dot{q}_{\textrm{m}}} = 1460~\text{rad s$^{-2}$}$, $s_{q_{\textrm{LP}}} = 6~\text{rad s$^{-1}$}$, $s_{\delta_{\textrm{LP}}} = 0.1\pi$, and $s_{M_{\textrm{LP}}} = |M_d| \cdot \delta_{\textrm{max}}\;\si{N\,m}$. For simplicity of notation, the normalized, augmented system will henceforth be denoted as
\begin{equation}
    \label{eq:NormalizedAugmentedSystem}
    \bar{x}^+ = \bar{f}(\bar{x}, \tilde{\varphi}(\bar{x})).
\end{equation}

\subsubsection{Semialgebraic Representation of the Closed-Loop System}
\label{sec:SemialgebraicSetModel}
To obtain the final \acrshort{SOS}-compatible system model, the nonlinear memoryless operator $\tilde{\varphi}(\tilde{x})$ and the closed-loop dynamics must be formulated as semialgebraic sets.

The exact piecewise-linear behavior of the actuator saturation is captured by representing the \acrshort{ReLU} activations via polynomial (in)equality constraints, as each \acrshort{ReLU} activation $\lambda_i = \max(0, v_i) \; \text{for } i\in\{1,2\}$ is described exactly by the semialgebraic set \cite{Korda2022, Detailleur2025}:
\begin{equation}
  \bigl\{(v_i,\lambda_i)\;\big|\;
  \lambda_i \ge 0,\;
  \lambda_i - v_i \ge 0,\;
  \lambda_i(\lambda_i - v_i) = 0
  \bigr\},
  \label{eq:ReLUset}
\end{equation}
with pre-activations $v_1 = \delta - \delta_{\textrm{min}}$ and $v_2 = \delta - \delta_{\textrm{max}}$.

The exact input-output relationship of the memoryless operator $\tilde{\varphi}$ is then expressed as the set:
\begin{equation}
\label{eq:SemialgebraicNetworkSet}
\mathbf{K}_{\tilde{\varphi}} = \Bigg\{ \bigg(\bar{x}, 
\begin{bmatrix} \lambda \\ \tilde{\varphi} \end{bmatrix} \bigg)
\; \Bigg| \; 
g^{\tilde{\varphi}}(\bar{x}, \lambda, \tilde{\varphi}) \geq 0, \;
h^{\tilde{\varphi}}(\bar{x}, \lambda, \tilde{\varphi}) = 0 \Bigg\},
\end{equation}
where $g^{\tilde{\varphi}}$ and $h^{\tilde{\varphi}}$ denote the vector-valued polynomials aggregating the inequalities and equalities of \cref{eq:ReLUset}, respectively. To simplify the notation in the remainder of this work, the variables spanning $\mathbf{K}_{\tilde{\varphi}}$ are grouped into the stacked vector $\zeta = [\bar{x}\transpose, \, \lambda\transpose, \, \tilde{\varphi}\transpose]\transpose \in \realsN{n_{\bar{x}}+n_\lambda+n_{\tilde{\varphi}}} = \realsN{n_\zeta}$.

Subsequently, to provide the \acrshort{SOS} program with the necessary information on how the closed-loop system evolves over time, the current state and activation $(\bar{x}, \tilde{\varphi})$ are coupled with the subsequent state and activation $(\bar{x}^+, \tilde{\varphi}^+)$. To compactly represent this coupling, the state-transition variables are defined as the vector $\xi = [\bar{x}\transpose, \lambda\transpose, \, \tilde{\varphi}\transpose, \, ({\bar{x}}^{+})\transpose, \, (\lambda^+)\transpose, \, ({\tilde{\varphi}^+})\transpose]\transpose \in \realsN{2(n_{\bar{x}} + n_\lambda + n_{\tilde{\varphi}}) - n_{\bar{x}}} = \realsN{n_\xi}$. This temporal evolution is then expressed as the semialgebraic set
\begin{equation}
    \label{eq:SemialgebraicComposedLoopSet}  
    \mathbf{K}_L =  \Bigg\{ \xi \in \realsN{n_\xi} \; \bigg\vert \;  g^L(\xi) \geq 0, \; h^L(\xi) = 0 \, \Bigg\},
\end{equation}
with
\begin{equation}
    g^L = \begin{bmatrix}
        g^{\tilde{\varphi}}(\zeta) \\
        g^{\tilde{\varphi}}(\zeta^+)
    \end{bmatrix}, \quad
    h^L = \begin{bmatrix}
        h^{\tilde{\varphi}}(\zeta) \\
        h^{\tilde{\varphi}}(\zeta^+) \\
        \bar{x}^+ - \bar{f}(\bar{x}, \tilde{\varphi}) 
    \end{bmatrix}.
\end{equation}
Together, $(\mathbf{K}_{\tilde\varphi}, \mathbf{K}_L)$ completely and exactly describe the saturated closed-loop system in a format compatible with \acrshort{SOS} programming.
\begin{rem}
    Following the methodology in \cite{Detailleur2025}, the semialgebraic sets $(\mathbf{K}_{\tilde\varphi}, \mathbf{K}_L)$ are augmented with redundant \textit{cross terms}, i.e. additional polynomial constraints derived from the products of the existing inequalities. The inclusion of these terms reduces the conservatism of the \acrshort{SOS} relaxations and facilitates the search for a feasible Lyapunov certificate.
\end{rem}

\subsection{Local Stability Analysis using Sum of Squares}
Once the exact semialgebraic model $(\mathbf{K}_{\tilde\varphi}, \mathbf{K}_L)$ is established, a local stability analysis with respect to the origin is conducted. To this end, a locally valid Lyapunov function and an inner estimate of the system's \acrshort{ROA} are sought over a predefined region $\bar{\mathcal{Q}}$.
\begin{defn}[Local Lyapunov Function, Def. B.12 \cite{Rawlings2017}]
\label{defn:LyapunovFunction}
Let $\mathcal{X} \subseteq \bar{\mathcal{Q}}$ be a positively invariant set containing the origin. A function $V: \mathcal{X} \to \mathbb{R}_{\geq 0}$ is defined as a local Lyapunov function for the system \labelcref{eq:NormalizedAugmentedSystem} if there exist functions $\alpha_1, \alpha_2 \in \mathcal{K}_{\infty}$ and a continuous positive-definite function $\alpha_3$ such that for all $\bar{x} \in \mathcal{X}$:
\begin{subequations}\label{eq:V_conditions}
    \begin{alignat}{1}
    V(\bar{x}) &\geq \alpha_1(\|\bar{x}\|) \label{eq:V_conditions_a}, \\
    V(\bar{x}) &\leq \alpha_2(\|\bar{x}\|) \label{eq:V_conditions_b}, \\
     V(\bar{x}) - V\big(\bar{f}(\bar{x},\tilde{\varphi}(\bar{x}))\big) &\geq \alpha_3(\|\bar{x}\|)\label{eq:V_conditions_c}.
\end{alignat}
\end{subequations}
\end{defn}
The \acrshort{ROA} is then defined as the set of all initial states $\bar{x}(0) \in \bar{\mathcal{Q}}$ such that $\lim_{k\to\infty} \|\bar{x}(k)\| = 0$. In this framework, the \acrshort{ROA} is estimated via the largest invariant sub-level set $\mathcal{L}_\gamma(V)$ of the Lyapunov function. The search is formulated as a sequence of two \acrshortpl{SDP} designed to identify a Lyapunov candidate and maximize its corresponding sub-level set, respectively. The consecutive feasibility of these optimization problems provides a formal certificate that the resulting polynomial $V(\bar{x})$ constitutes a locally valid Lyapunov function for the closed-loop system \cite{Detailleur2025}.

The remainder of this section is structured to reflect the steps of the stability verification procedure. First, \cref{sec:SearchRegionQ} describes the determination of the bounded search region $\bar{\mathcal{Q}}$. Subsequently, the search for the Lyapunov candidate function $V(\bar{x})$ is detailed in \cref{sec:SearchForALocalLyapunovCandidateFunction}. Finally, \cref{sec:SearchForLargestSublevelSet} presents the optimization of the largest sub-level set $\mathcal{L}_\gamma(V)$.

\subsubsection{Defining the Bounded Search Region $\bar{\mathcal{Q}}$}
\label{sec:SearchRegionQ}
To synthesize a locally valid Lyapunov function, it is necessary to define a bounded search region within which the Lyapunov candidate function must strictly decrease. This region is defined via the normalized augmented state-space as the semialgebraic set
\begin{equation}
    \bar{\mathcal{Q}} = \big\{ \bar{x} \in \realsN{n_{\bar{x}}} \mid \bar{q}(\bar{x}) \geq 0 \big\},
    \label{eq:LocalRegionQDefinition}
\end{equation}
such that the origin $0 \in \textrm{int}(\bar{\mathcal{Q}})$. The bounding function is chosen to be strictly quadratic in the normalized augmented states, defined as
\begin{equation}
    \bar{q}(\bar{x}) = \alpha_{\bar q} - \bar{x}\transpose \bar{Q} \bar{x},
\end{equation}
where $\alpha_{\bar q} > 0$ is a tuning parameter that determines the overall volume of the search domain.

To improve the numerical tractability of the \acrshort{SOS} optimization, the shape of the search region must be carefully initialized. As established by \cref{cor:AugmentedStateFeedbackController}, the augmented system behaves as a saturated state-feedback controller, thus, the bounding matrix $\bar{Q}$ is parameterized as $\bar{Q} = P_\ell$, where $P_\ell \succ 0$ is the solution to the discrete Lyapunov equation
\begin{equation}
  \bar{A}_{cl}\transpose P_\ell\,\bar{A}_{cl} - P_\ell = -I,
  \label{eq:discreteLyapunov}
\end{equation}
with $\bar{A}_{cl} = D_x^{-1} ( \tilde{A} + \tilde{B}\tilde{K}) D_x$ representing the normalized, linearized closed-loop system matrix.

\subsubsection{Search for a Local Lyapunov Candidate Function $V$}
\label{sec:SearchForALocalLyapunovCandidateFunction}
In the next step, the first \acrshort{SDP} is formulated to obtain locally valid Lyapunov candidate function. To reduce conservatism in the stability analysis, the candidate function is not restricted to depend solely on the system states. Instead, it is constructed as a polynomial function of the augmented state, the hidden network activation functions, and the memoryless operator outputs. Hence, the Lyapunov candidate function is parameterized by
\begin{equation}
{\scalebox{0.90}{$
\begin{aligned}
    \label{eq:LyapunovParam}
        V(\bar{x}) &= V_{\zeta}(\bar{x}, \lambda(\bar{x}), \tilde{\varphi}(\bar{x}))\\ &
        \begin{multlined}
            = \sigma^V\big(\bar{x}, \lambda(\bar{x}), \tilde{\varphi}(\bar{x})\big) \\
            + \sigma^V_{\textrm{ineq}}(\bar{x}, \lambda(\bar{x}), \tilde{\varphi}(\bar{x}))\transpose
            \underbrace{\begin{bmatrix}
                \mathcal{M}\big(g^{\tilde{\varphi}}(\bar{x}, \lambda(\bar{x}), \tilde{\varphi}(\bar{x})), 1\big) \\
                \mathcal{M}(g^{\tilde{\varphi}}(\bar{x}, \lambda(\bar{x}), \tilde{\varphi}(\bar{x})), 2\big) \\
                \vdots \\
                \mathcal{M}\big(g^{\tilde{\varphi}}(\bar{x}, \lambda(\bar{x}), \tilde{\varphi}(\bar{x})), n\big)
                \end{bmatrix}}_{g^V},
                \end{multlined}
\end{aligned}
$}}
\end{equation}
with parameters $\sigma^V$, $\sigma^V_{\textrm{ineq}}$ representing any scalar \acrshort{SOS} polynomial and any vector of \acrshort{SOS} polynomials, respectively. Here $n$ can be chosen to increase the expressivity of the candidate Lyapunov function, and thus enlarge the solution space, at the cost of increased computational complexity.

A valid local Lyapunov candidate function must strictly decrease along the system trajectories within the defined search region $\bar{\mathcal{Q}}$. To enforce this decrease condition locally, the semialgebraic description $\mathbf{K}_L$ of the composed loop is bounded by incorporating the polynomial inequality $\bar{q}(\bar{x}) \ge 0$. Applying the Positivstellensatz yields the condition \cite{Detailleur2025, Parrilo2003, Korda2022}
\begin{equation}
{\scalebox{0.90}{$
\begin{aligned}
    & V_\zeta(\zeta) - V_\zeta(\zeta^+) - \|\bar{x}\|_P^2 -  p^{\Delta V}_{\textrm{eq}}(\xi)\transpose h^L(\xi)
    \\
    & \ \  - \sigma^{\Delta V}(\xi) - \sigma^{\Delta V}_{\textrm{ineq}}(\xi) \transpose 
     \begin{bmatrix}
        \mathcal{M}\Big( \left[\!\begin{smallmatrix} \bar{q}(\bar{x}) \\ g^L(\xi) \end{smallmatrix}\!\right]^{\phantom{|}}\!\!, 1 \Big)
        \\
        \mathcal{M}\Big( \left[\!\begin{smallmatrix} \bar{q}(\bar{x}) \\ g^L(\xi) \end{smallmatrix}\!\right]^{\phantom{|}}\!\!, 2 \Big)
        \\
        \vdots \\
        \mathcal{M}\Big( \left[\!\begin{smallmatrix} \bar{q}(\bar{x}) \\ g^L(\xi) \end{smallmatrix}\!\right]^{\phantom{|}}\!\!, n \Big)
    \end{bmatrix} 
    \geq 0, \  \forall \xi \in \realsN{n_{\xi}},
    \label{eq:LocalLyapunovDecreaseCond}
\end{aligned}
$}}
\end{equation}
with weighting matrix $P$. Note that the superscript $\Delta V$ is chosen to indicate the membership of a variable to the decrease condition. This condition is interpreted as an \acrshort{SOS} constraint in the \acrshort{SDP}
\begin{subequations}
    \label{eq:SDPFormulationLocalAsymptoticStability}
    \begin{alignat}{4}
        &\span\span \text{find:} \ & P, \, \sigma^V, \, \sigma^V_{\textrm{ineq}}, \ \ \, \,
        & \!\!\! \! \! \! \! \sigma^{\Delta V},  \, \sigma^{\Delta V}_{\textrm{ineq}}, \,  p^{\Delta V}_{\textrm{eq}} \span \nonumber \\
        &\span\span \text{s.t.} \ & \eqref{eq:LyapunovParam}, \ & \eqref{eq:LocalLyapunovDecreaseCond}, \span \\
        &\span\span & \sigma^V, \sigma^V_{\textrm{ineq}}, \sigma^{\Delta V}, \sigma^{\Delta V}_{\textrm{ineq}} \quad & \text{\acrshort{SOS} polynomials,} \label{eq:LocalAsymptoticStabilitySOS} \\
        &\span\span & 
        p^{\Delta V}_{\textrm{eq}} \quad & \text{arbitrary polynomials,} \label{eq:LocalAsymptoticStabilityP} \\
        &\span\span & P \quad & \succ 0. \label{eq:LocalAsymptoticStabilityMatrixP}
    \end{alignat}
\end{subequations}
Because $\bar{\mathcal{Q}}$ is not known a priori to be positively invariant, a general solution to \acrshort{SDP} \labelcref{eq:SDPFormulationLocalAsymptoticStability} is not guaranteed to define a valid Lyapunov function.

\subsubsection{Search for the Largest Lyapunov Sub-Level Set $\mathcal{L}_\gamma(V)$}
\label{sec:SearchForLargestSublevelSet}
Once a solution to \acrshort{SDP} \labelcref{eq:SDPFormulationLocalAsymptoticStability} has been obtained, certifying that this solution defines a valid Lyapunov function for the closed-loop system requires finding a positively invariant set $\mathcal{X} \subseteq \bar{\mathcal{Q}}$, which concurrently forms an inner estimate of the system's \acrshort{ROA}. By \cref{eq:LocalLyapunovDecreaseCond}, a trivial class of invariant sets within $\bar{\mathcal{Q}}$ consists of the sublevel sets of $V$, defined as $\mathcal{L}_\gamma(V) \coloneq \big\{ \bar{x} \in \realsN{n_{\bar{x}}} \mid V(\bar{x}) \leq \gamma \big\}$. Therefore, a second optimization problem is formulated to find the largest sublevel set of $V$ contained entirely in $\bar{\mathcal{Q}}$.

To restrict the search to sublevel sets that are strictly contained within the set $\bar{\mathcal{Q}}$, i.e. $\mathcal{L}_\gamma(V) \subseteq \bar{\mathcal{Q}}$, we consider the condition
\begin{equation}
{\scalebox{0.90}{$
    \begin{aligned}
        \sigma^{\mathcal{Q}}_q(\zeta) \bar{q}(\bar{x}) &\geq p^{\mathcal{Q}}_{\textrm{eq}}(\zeta)\transpose h^{\tilde{\varphi}}(\zeta) + \sigma^{\mathcal{Q}}(\zeta) + \sigma^{\mathcal{Q}}_V(\zeta) \big( \gamma - V_\zeta(\zeta) \big) \\ 
        &\quad + \sigma^{\mathcal{Q}}_{\textrm{ineq}}(\zeta)\transpose  \begin{bmatrix}
            \mathcal{M}\big(g^{\tilde{\varphi}}(\zeta), 1 \big) \\
            \mathcal{M}\big(g^{\tilde{\varphi}}(\zeta), 2 \big) \\
            \vdots        \\
            \mathcal{M}\big(g^{\tilde{\varphi}}(\zeta), n \big)
        \end{bmatrix}, \  \forall \zeta \in \realsN{n_{\zeta}},
    \end{aligned}
$}}
\label{eq:SublevelSetCond}
\end{equation}
with $p^{\mathcal{Q}}_{\textrm{eq}}$ representing a vector of arbitrary polynomials, $\sigma^{\mathcal{Q}}_q$, $\sigma^{\mathcal{Q}}_V$ and $\sigma^{\mathcal{Q}}$ representing scalar \acrshort{SOS} polynomials, and $\sigma^{\mathcal{Q}}_{\textrm{ineq}}$ representing a vector of \acrshort{SOS} polynomials. The function $V_\zeta$ is fixed to the solution obtained from \acrshort{SDP} \labelcref{eq:SDPFormulationLocalAsymptoticStability}. Applying the Positivstellensatz over the static network constraints of $\mathbf{K}_{\tilde{\varphi}}$, it follows that the satisfaction of \cref{eq:SublevelSetCond} is a sufficient condition to guarantee $\bar{q}(\bar{x}) \geq 0$ for all $\bar{x} \in \mathcal{L}_{\gamma}(V)$. 

\begin{figure*}[t]
    \centering
    \begin{subfigure}[t]{0.52\textwidth}  
        \centering
        \centering
        \includegraphics[width=\linewidth,
        keepaspectratio]{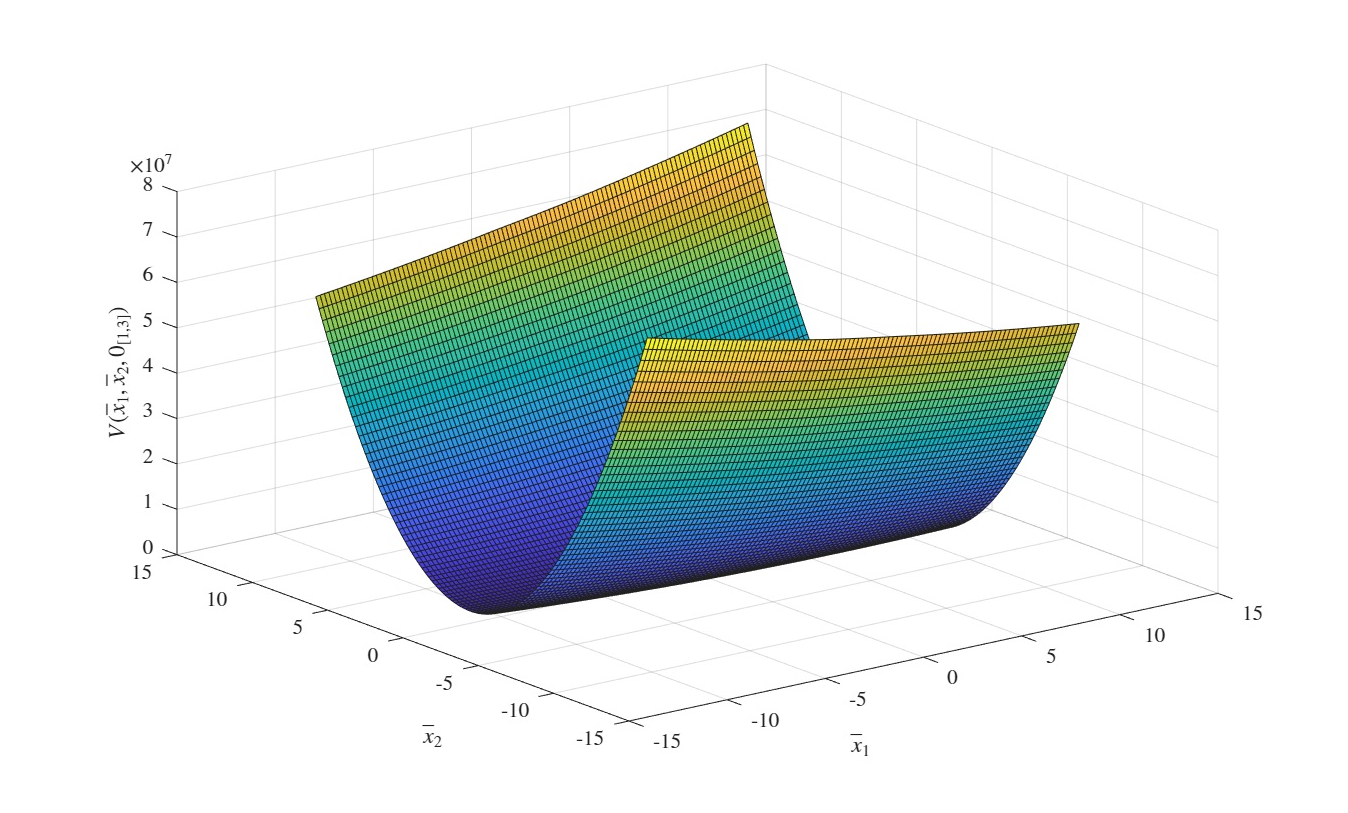}
        \caption{Surface plot of the \acrshort{SOS}-returned polynomial $V_\zeta(\bar{x})$ on the slice $\bar{x}_\varphi = 0$. The function is non-negative over the domain shown and has its minimum at the origin. 
        }
        \label{fig:LyapunovSurface}
    \end{subfigure}
    \hfill
    \begin{subfigure}[t]{0.47\textwidth}  
        \centering
        \includegraphics[width=\linewidth]{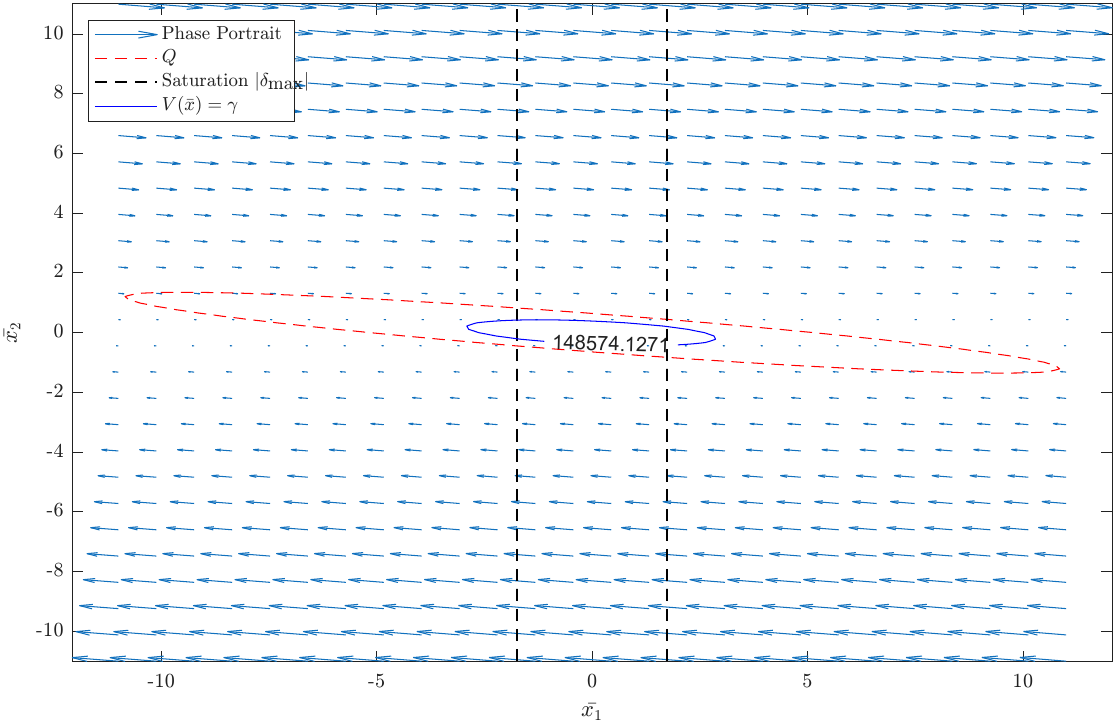}
        \caption{Phase portrait of the saturated \acrshort{INDI} closed-loop on the slice $\bar{x}_\varphi = 0$ in normalized coordinates $(\bar{x}_1, \bar{x}_2)$. Solid blue: boundary of the certified sublevel set $\mathcal{L}_\gamma(V)$. Red dashed: boundary of the region $\bar{\mathcal{Q}}$ with $\alpha_{\bar q} = 1000$. Black dotted: saturation limits $\pm\delta_{\textrm{max}}$ projected onto the $(\bar{x}_1, \bar{x}_2)$ plane.} 
        \label{fig:PhasePortrait}
    \end{subfigure}

    \caption{Lyapunov function found via an \acrshort{SOS}-based stability verification approach.} 
    \label{fig:StabilityAnalysis}
\end{figure*}

Viewing \cref{eq:SublevelSetCond} as an \acrshort{SOS} constraint leads to the optimization problem
\begin{subequations}
    \label{eq:SDPFormulationLargestSublevelSet}
    \begin{alignat}{4}
        &\span\span  \underset{ \substack{ \gamma, \; \sigma^{\mathcal{Q}}_q, \, \sigma^{\mathcal{Q}}, \\ \sigma^{\mathcal{Q}}_V, \, \sigma^{\mathcal{Q}}_\textrm{ineq}, \, p^{\mathcal{Q}}_{\textrm{eq}}}}{\textrm{maximize:}} \ &   & \! \! \! \! \! \! \! \! \gamma  \span \nonumber \\
        &\span\span \text{s.t.} \ &  & \! \! \! \! \! \! \! \! \! \! \! \! \!   \eqref{eq:SublevelSetCond}, \span \label{eq:SDPFormulationLargestSublevelSet_SublevelSetCond} \\
        &\span\span & \sigma^{\mathcal{Q}}_q, \sigma^{\mathcal{Q}}, \sigma^{\mathcal{Q}}_V, \sigma^{\mathcal{Q}}_{\textrm{ineq}} \quad & \text{\acrshort{SOS} polynomials,} \label{eq:LargestSublevelSetSOS} \\
        &\span\span & p^{\mathcal{Q}}_{\textrm{eq}} \quad & \text{arbitrary polynomials,} \label{eq:LargestSublevelSetP}
    \end{alignat}
\end{subequations}
which is solved as an \acrshort{SDP} by fixing $\sigma^{\mathcal{Q}}_V$ to a strictly positive \acrshort{SOS} polynomial. 

If consecutively solving \acrshortpl{SDP} \labelcref{eq:SDPFormulationLocalAsymptoticStability,eq:SDPFormulationLargestSublevelSet} yields a solution, this solution with $V_\zeta(0) < \gamma$ defines a valid local Lyapunov function \cite{Detailleur2025}, formally proving that $\mathcal{X} = \mathcal{L}_\gamma(V)$ forms an inner estimate of the closed-loop system's \acrshort{ROA}.

\section{Results}
\label{sec:Results}
In this section, the numerical results of the stability verification procedure applied to the saturated \acrshort{INDI} pitch-rate loop are presented. The nominal controller and system parameters at the $V_a = 20$m/s fixed-wing trim condition are listed in \cref{tab:nominalParameters} and are taken from\cite{Schlatter2024, Schenk2026}. All matrices are discretized using an forward-Euler discretization with a sampling period $T_s = 0.01~\text{s}$. All \acrshortpl{SDP} are solved using MOSEK \cite{mosek} to a numerical tolerance of $10^{-9}$. The Lyapunov candidate $V$ and the decrease condition $\sigma^{\Delta V}$ are chosen as degree-4 \acrshort{SOS} polynomials. All associated multipliers and equality-constraint polynomials are set to degree 2. Additionally, second degree cross-terms are considered to facilitate the search for a Lyapunov function, yielding a $12$ dimensional vector $g^{\tilde\varphi}$. Since $\alpha_{\bar q}$ defines the volume of the set $\bar{\mathcal{Q}}$, it is set to $\alpha_{\bar q} = 1000$. Further increasing this value prevents the solver from finding a solution.

\begin{table}[t]
  \caption{Nominal parameter values at the $V_a = 20~\text{m/s}$
           fixed-wing trim condition\cite{Schlatter2024,Schenk2026}.}
  \label{tab:nominalParameters}
  \centering
  \begin{tabular}{llll}
    \toprule
    Symbol & Value & Unit & Origin \\
    \midrule
    \multicolumn{4}{l}{\textit{System parameters}} \\
    $I_{yy,m}$  & $0.025$ & $\si{kg\,m^2}$      & physical airframe \\
    $\tau_{\textrm{act}}$ & $0.05$  & $\si{s}$             & servo model \\
    $M_m$        & $-16.7982$  & $\si{N\,m\,rad^{-1}}$ & trim, \cite{Schlatter2024} \\
    $d_q$        & $0.196$  & $\si{N\,m\,s\,rad^{-1}}$ & trim, \cite{Schlatter2024} \\
    $M_{\alpha}$ & $10.8515$ &$\si{N\,m\,rad^{-1}}$ & trim,  \cite{Schlatter2024} \\
    \midrule
    \multicolumn{4}{l}{\textit{Controller parameters (nominal = matched)}} \\
    $I_{yy,c}$  & $I_{yy,m}$ & $\si{kg\,m^2}$      & matched \\
    $M_d$        & $M_m$      & $\si{N\,m\,rad^{-1}}$ & matched \\
    $k_p$        & $20$       & $\si{rad\,s^{-1}}$  & \cite{Schlatter2024} \\
    $\tau_q$     & $0.04$    & $\si{s}$             & estimator corner \\
    $k_{\textrm{LP}}$     & $25$      & $\si{rad\,s^{-1}}$  & $= 1/\tau_q$ \\
    $\tau_\delta$ & $0.05$   & $\si{s}$             & $= \tau_{\textrm{act}}$ \\
    $\delta_{\textrm{max}}$ & $2\pi/9$ & $\si{rad}$          & \cite{Schlatter2024} \\
    \bottomrule
  \end{tabular}
\end{table}

Solving \acrshortpl{SDP} \cref{eq:SDPFormulationLocalAsymptoticStability,eq:SDPFormulationLargestSublevelSet} yields the inner estimate of the \acrshort{ROA} $\mathcal{X} = \mathcal{L}_\gamma(V)$, which is visualized in \Cref{fig:PhasePortrait} for $\bar{x}_\varphi = 0$. Notably, despite the open-loop instability of the plant, the certified region reaches into the projected saturation limits $\pm\delta_{\textrm{max}}$, confirming that the certificate accounts for active actuator saturation. 

\section{Discussion \& Conclusion}
\label{sec:Conclusion}
To provide a formal stability certificate for the saturated \acrshort{INDI} architecture, this work demonstrated that the saturated \acrshort{INDI} controller and its state estimator are structurally equivalent to a \acrshort{REN}. By leveraging this equivalence alongside a closed-loop augmented state-feedback formulation, \acrshort{SOS} programming was used to synthesize a locally valid Lyapunov function and a certified estimate of the \acrshort{ROA} that explicitly accounts for active actuator saturation. To the best of the authors' knowledge, these results provide the first formal nonlinear stability certificate for \acrshort{INDI} controllers, showcasing the potential of \acrshort{SOS}-based stability verification for flight control architectures.

By establishing this foundational equivalence, this work paves the way for advanced nonlinear stability analyses of \acrshort{INDI} control for hybrid \acrshort{VTOL} aircraft. Building upon these results, future research should address expanding the certified \acrshort{ROA} by implementing the iterative alternating sequence of \acrshortpl{SDP} to directly optimize the region, as proposed in \cite{Detailleur2025}. Furthermore, to address model mismatch, the framework will be extended to provide robust stability certificates by incorporating sector and slope constraints on the control effectiveness. Finally, the analysis will be scaled to full nonlinear attitude dynamics to account for the cross-coupling effects inherent in high-performance maneuvers.

\section*{References}
\bibliographystyle{IEEEtranBST/IEEEtran}
\renewcommand{\section}[2]{}%
\bibliography{IEEEtranBST/IEEEabrv,mybibfile}
\end{document}